\documentclass[letterpaper, 10 pt, conference]{ieeeconf}  % Comment this line out if you need a4paper

\IEEEoverridecommandlockouts                              % This command is only needed if 
\usepackage{graphics} % for pdf, bitmapped graphics files
\usepackage{graphicx}
\usepackage{cite}
\usepackage{epsfig} % for postscript graphics files
\usepackage{times} % assumes new font selection scheme installed
\usepackage[cmex10]{amsmath} % assumes amsmath package installed
\usepackage{mathtools,bbm}
\usepackage{xfrac}
\usepackage{algorithm}
\usepackage{algpseudocode}
\usepackage{verbatim}
\usepackage{tabularx}
\usepackage{booktabs}
\usepackage{multirow}
\usepackage{adjustbox}
\usepackage{multirow}
\usepackage{float}
\let\labelindent\relax
\usepackage{enumitem}
\usepackage{amsfonts}
\usepackage[dvipsnames]{xcolor}
\usepackage{standalone}
\usepackage{forest}
\usepackage{calc}
\newcommand{\rvec}{\bm{r}}

\newcommand{\yvec}{\bm{y}}
\newcommand{\uvec}{\bm{u}}

\newcommand{\evec}{\bm{e}}

\usepackage{tikz}
\usetikzlibrary{shapes.geometric, arrows}

\tikzstyle{startstop} = [rectangle, rounded corners, minimum width=1cm, minimum height=0.75cm,text centered, draw=black, fill=green!5]
\tikzstyle{io} = [trapezium, trapezium left angle=70, trapezium right angle=110, minimum width=3cm, minimum height=1cm, text centered, draw=black, fill=orange!15]
\tikzstyle{process} = [rectangle, minimum width=2cm, minimum height=0.75cm, text centered, text width=2.5cm, draw=black, fill=orange!15]
\tikzstyle{process2} = [rectangle, minimum width=1cm, minimum height=1cm, text centered, draw=black, fill=red!20]
\tikzstyle{processLong} = [rectangle, minimum height=0.75cm, text centered,  draw=black, fill=orange!15]
\tikzstyle{decision} = [diamond, minimum width=2cm, minimum height=0.75cm, text centered, draw=black, fill=blue!15]
\tikzstyle{arrow} = [thick,->,>=stealth]
\tikzstyle{sum} = [draw, circle, minimum size=.25cm]

\newcommand{\BEAS}{\begin{eqnarray*}}
	\newcommand{\EEAS}{\end{eqnarray*}}
\newcommand{\BEA}{\begin{eqnarray}}
	\newcommand{\EEA}{\end{eqnarray}}
\newcommand{\BEQ}{\begin{equation}}
	\newcommand{\EEQ}{\end{equation}}
\newcommand{\BIT}{\begin{itemize}}
	\newcommand{\EIT}{\end{itemize}}

\newcommand{\reals}{{\mbox{$\mathbb{R}$}}}
\newcommand{\integers}{{\mbox{$\mathbb{Z}$}}}

\newcommand{\argmin}{\mathop{\rm argmin}}

\usepackage{amsthm}
\newtheorem{theorem}{Theorem}[]

\newtheorem{proposition}[theorem]{Proposition}

\usepackage{bm}
\usepackage{color}
\usepackage{amssymb}
\newcommand{\vvec}{\bm{v}}

\usepackage[normalem]{ulem}

\newcommand{\revn}[1] {\textcolor{.}{#1}}

\makeatletter
\newcommand*\titleheader[1]{\gdef\@titleheader{#1}}
\AtBeginDocument{%
	\let\st@red@title\@title%
	\def\@title{%
		\bgroup\normalfont\large\centering\@titleheader\par\egroup
		\vskip1.5em\st@red@title}
}
\makeatother

\titleheader{
	\small{\textcopyright 2021 IEEE.  Personal use of this material is permitted.  Permission from IEEE must be obtained for all other uses, in any current or future media, including reprinting/republishing this material for advertising or promotional purposes, creating new collective works, for resale or redistribution to servers or lists, or reuse of any copyrighted component of this work in other works.}\\
	This is the accepted version of the paper:
	E. C. Balta, K. Barton, D. M. Tilbury, A. Rupenyan, J. Lygeros, ``Learning-Based Repetitive Precision Motion Control with Mismatch Compensation,'' 60th IEEE Conference on Decision and Control, Austin, Texas, USA, 2021.}

\title{
	{Learning-Based Repetitive Precision Motion Control with\\ Mismatch Compensation}
}

\author{Efe C. Balta, Kira Barton, Dawn M. Tilbury, Alisa Rupenyan, John Lygeros 
\thanks{
This project has been funded in part by the Swiss National Science Foundation under NCCR Automation and in part by NSF Award \# 1544678.
	} 
	\thanks{
		K. Barton, and D. M. Tilbury are with the Department of Mechanical Engineering, University of Michigan, Ann Arbor, MI 48109, USA
		{\tt\small$\{$tilbury,bartonkl$\}$@umich.edu},}
	\thanks{
		E. C. Balta, A. Rupenyan and J. Lygeros are with the Automatic Control Laboratory, ETH Zurich, 8092 Zurich, Switzerland. A. Rupenyan is also with Inspire AG, 8092
		Zurich, Switzerland {\tt\small$\{$ebalta,ralisa,lygeros$\}$@control.ee.ethz.ch}}
}

\begin{document}
	
%	\maketitle
text asdfasdf
	{\let\newpage\relax\maketitle}
	\begin{abstract}
		Learning-based control methods utilize run-time data from the underlying process to improve the controller performance under model mismatch and unmodeled disturbances.
		This is beneficial for optimizing industrial processes, where the dynamics are difficult to model, and the repetitive nature of the process can be exploited.
		In this work, we develop an iterative approach for repetitive precision motion control problems where the objective is to follow a reference geometry with minimal tracking error.
		Our method utilizes a nominal model of the process and learns the mismatch using Gaussian Process Regression (GPR).
		The control input and the GPR data are updated after each iteration to improve the performance in a run-to-run fashion.
		We provide a preliminary convergence analysis, implementation details of the proposed controller for minimizing different error types, and a case study where we demonstrate improved tracking performance with simulation and experimental results. 
	\end{abstract}
	
	\section{Introduction}
	
	Learning-based control methods have been increasingly utilized in academia and industry to effectively incorporate data-driven learning methods into intelligent control frameworks. 
	In most applications, online feedback is used in conjunction with a data-driven method to improve (or provide~\cite{umlauft2018uncertainty}) a plant model, which is then used by the controller to compute the next control action. 
	The learning may occur simultaneously with the control evaluation or in sequential runs (i.e., iterations), depending on the system architecture and computational availability.

	In this work, 
	we focus on repetitive precision motion control tasks \revn{and learning in sequential iterations of the process}, where a motion stage is tasked with tracking a given reference with the minimum error possible (see~\cite{oomen2018advanced} for a literature survey). 
	Many precision motion control applications leverage the repetitive nature of the task to develop run-to-run learning controllers that utilize the system response in a previous run to improve the control signal in the next run, \revn{where the control signal cannot be adjusted during runs}. 
	Iterative Learning Control (ILC) has been extensively applied for many repetitive precision motion control problems in the literature due to its high performance \cite{wang2017newton,altin2014robust}. ILC uses the nominal model of an uncertain plant and provides robustness to disturbances and model mismatches by leveraging online measurements from the plant~\cite{amann1998predictive}. 
	Recent related work on run-to-run control includes \cite{colombino2019towards}, where the authors use a bounded linear approximation of a power system to derive robust convergence to an approximated fixed point.
	In~\cite{baumgartner2020zero}, the authors have adopted a zeroth-order approximation of a generic nonlinear plant response from a previous run to show local convergence (under bounded modeling error). In all the aforementioned works, robust convergence often results in suboptimality as a function of the nominal model mismatch.
	Additionally, the measurements are not leveraged for improving the plant model itself for improved performance. 
	
	\revn{Learning-based tracking control is applied to many simultaneous (online) learning applications, including adaptive control~\cite{cheah2006adaptive}, reinforcement learning~\cite{modares2014optimal}, and Gaussian state-space models~\cite{beckers2019stable,umlauft2018uncertainty}. 
	However, these methods often rely on run-time adjustments of control inputs, which may be impractical for the processes of interest in this work.}
	\revn{Additionally}, learning-based control methods have been proposed for modeling the mismatch between a nominal plant model and the true plant~\cite{hewing2019cautious,berkenkamp2015safe}. 
	By utilizing a data-driven representation of the mismatch between the nominal model and the true plant, improved performance of the learning-based controllers has been demonstrated in multiple applications. However, the mismatch modeling approach has not been applied to the run-to-run control problem to augment the nominal model and improve controller performance. \cite{li2020data} develops a data-driven gradient estimation scheme for repetitive precision motion control, but the proposed framework does not have formal stability guarantees.

	In this work, we propose a novel learning-based control framework that employs measurements from a previous run to augment the nominal plant model and provide convergence for an iterative precision motion control problem. 
	The main contributions are:
	\begin{itemize}[leftmargin=*]
		\item A new data-driven compensation method to improve the tracking accuracy of repetitive motion processes.
		\item A preliminary analysis of convergence \sout{and robustness} of the proposed method in the iteration domain.
		\item A case study to illustrate the utility of the proposed work and preliminary experimental implementation.
	\end{itemize}
	\revn{The rest of the paper is structured as follows.
	Section~\ref{sec:prelim} provides the preliminaries.
	Section~\ref{sec:controller} presents the proposed run-to-run (R2R) controller.
	Section~\ref{sec:caseStudy} provides a simulation case study and preliminary experimental results.
	Section~\ref{sec:conc} provides closing remarks and future directions.}
	
	\section{Preliminaries}
	\label{sec:prelim}
	
	\subsection{Problem Setting}
	Figure~\ref{fig:testbed} presents the control framework for the two-axis precision motion experimental testbed with
	the two axes (X-Y) and the tool tip (TT), which has the input-output dynamics (\ref{eq:dynSys}).
	We have a reference geometry, denoted by $r$, e.g., the octagonal geometry illustrated in Fig.~\ref{fig:testbed}, that we would like the TT to follow with minimal error. 
	We assume that the process is repetitive in that the system is expected to track the same reference geometry in subsequent \textit{runs}.
	The true output map of the precision motion system is modeled by a known nominal model and an initially unknown perturbation term.
	The \textit{nominal} system model is taken as linear time invariant 
	\begin{align}
		\label{eq:linSys}
		x_k(t+1) &= A x_k(t) + B u_k(t), 
	\end{align}
	where $t = 1,\ldots,n_k$ is the discrete time index of length $n_k\in\integers_+$,
	$k=1,2,\ldots$ is the iteration index, $x_k(t)\in\reals^{n_x}, u_k(t)\in\reals^{n_u}$ are the nominal state vector and control input of iteration $k$, respectively, and $A,B$ are appropriately sized matrices. 
	We assume fixed iteration length of $n_i\in\integers_+$ for simplicity, e.g. $n_k = n_i$ for all $k$.
	Then, the \textit{nominal-state-to-output measurement map} of the system is given by
	\begin{align}
	\label{eq:dynSys}
	y_k(t+1) = h(x_k(t), u_k(t)) + g(x_k(t), u_k(t)) + w(t),
	\end{align}
	where $y_k(t)\in\reals^{n_y}$ is the measured output at iteration $k$, $g$ is the unknown systematic perturbation to the output measurement map, parametrized by the nominal state and the input, $h$ is the nominal output model, and $w(t)$ denotes the noise term, which is assumed to be i.i.d. with the distribution $\textstyle{w(t)\sim \mathcal{N}(0, \Sigma^w)}$ and a diagonal $\Sigma^w$ to facilitate the proposed mismatch learning method. 
	Please note that we use the terms \textit{run} and \textit{iteration} interchangeably. %throughout the paper.
	
	We assume that $h$ and $g$ are iteration invariant. The nominal linear output model is given by $h(x_k(t),u_k(t)) = C(Ax_k(t)+Bu_k(t))$. 
	The systematic (and potentially nonlinear) perturbation $g$, entering the output map, is initially unknown and we aim to learn $g$ from output measurement corrupted by random noise. 
	This model structure is similar to Wiener systems in the literature~\cite{bai2008towards}. 
	By assuming the unknown perturbation is part of the output map (\ref{eq:dynSys}) and not the nominal system model (\ref{eq:linSys}), we simplify model complexity when predicting the true output by leveraging the nominal input-output map of the lifted system given below, since the future output predictions do not depend on the current output predictions. 
	The lifted input-output mapping is
	\begin{align}
	\label{eq:io_map}
	\bm{y}_k 
	=\! H\uvec_k + \bar{H}x(0) + G(\uvec_k, x(0)) + \bm{w},
	\end{align} 
	where $\textstyle{\bm{y}_k = [y_k(1)^T, \ldots, y_k(n_i)^T ]^T}$, \revn{$\bm{w}$ is defined similarly,} $\uvec_k = [u_k(0)^T, \ldots, \textstyle{u_k(n_i\!-\!1)^T} ]^T$, 
	\vspace{-0.1cm}
	\begin{align*}
	\small H
	\!=\! 
	\begin{bmatrix}
	CB	 &\!  & 0	\\
	\vdots & \ddots  &   \\
	CA^{n_i-1}B &\ldots\! &CB
	\end{bmatrix}\!,
	\bar{H}\!=\!
	\begin{bmatrix}
	CA\\
	\vdots\\
	CA^{n_i}
	\end{bmatrix},
	\end{align*}
	and $G(\uvec_k, x(0))$ denotes the output mismatch map as the concatenation of $g(x_k(t), u_k(t))$ for the lifted form.
	We assume that the initial condition $x_*(0)$ is the same for all runs, assumed to be zero to simplify the notation.
	
	\begin{figure}[t]
		\centering
		\includegraphics[width=\columnwidth]{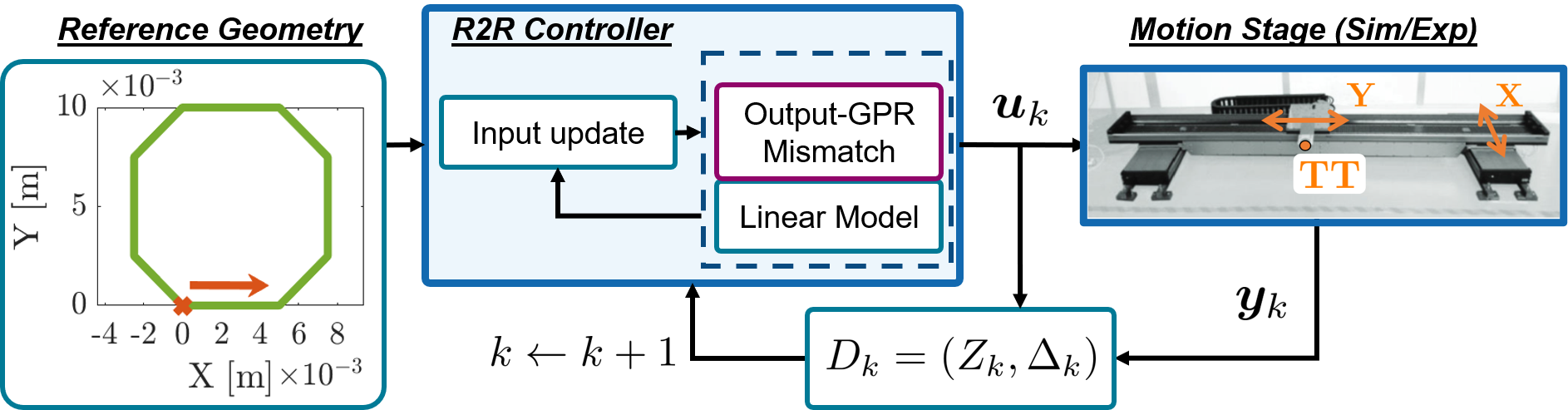}
%		\vspace{-0.2cm}
		\caption{Block diagram of the control setting in the case study. The experimental system block is replaced by the ``true'' system model for the simulation studies.}
		\label{fig:testbed}
		\vspace{-0.65cm}
	\end{figure}
	
	The \textit{control objective} is to improve the precision of the two dimensional (2D) motion stage by utilizing run-to-run (R2R) controller updates to have $\yvec_k \simeq \rvec$ as $k\to\infty$.
	We denote the parametrization of the reference trajectory with respect to its arc-length $s\in[0,L_r]$ as $r(s)$ and the 
	discretized representation of the reference signal as $\rvec\in\reals^{n_y}$.

	Several ways of quantifying tracking performance have been considered in the literature \cite{koren1991variable} 
	Here, two types of tracking error are considered. 
	The individual axis error for the discretized trajectory is defined as $e_k(j) = \rvec(j) - y_k(j)$,
	where we have $e_k(j) = [e_k^1(j), e_k^2(j)]$ for the errors of the two axes.
	The contouring error is defined as the minimum distance between an output point $y_k(j)$ and the reference trajectory $r(s)$.
	Let $\hat{s} = \argmin_s ||r(s) - y_k(j)||$ denote the curve-length that corresponds to the point on the reference closest to a current position $y_k(j)$.
	Then, the contouring error $\hat{e}^c(j)$ is defined as the distance between $r(\hat{s})$ and $y_k(j)$.

	\subsection{Gaussian Process Regression}
	\label{sec:gpr}
	In this work, we leverage the repetitive nature of the run-to-run process to learn the unknown system dynamics $g(x_k(t), u_k(t))$ from the input and output data of the system. 
	We apply Gaussian Process Regression (GPR) for modeling the output mismatch. 
	GPR can provide accurate estimates with small datasets, and uncertainty bounds on the posterior predictions, which we exploit in the proposed controller.
	
	Our dataset includes features of nominal state-input pairs, i.e., $(x^j, u^j)$, and noisy mismatch observations 
	\begin{align}
		\label{eq:delta}
	\delta^j = y^{j+1} - h(x^j, u^j) = g(x^j, u^j) + w^j.
	\end{align}
	Denote the feature vector as $\textstyle{z^j\!:=\! [(x^j)^T, (u^j)^T]^T}$, then the dataset is given by
	\begin{align*}
	D = \{ (Z, \Delta) ~|~ Z = [z^1,\ldots,z^{n_g}]^T, \Delta = [\delta^1,\ldots,\delta^{n_g}]^T\}.
	\end{align*}
	We call $Z$, the set of \textit{features} of the dataset, and $\Delta$, the \textit{observations}.
	Since the unknown function $g$ is vector-valued, we utilize multiple GPRs, one per output channel.
	Let $m\in\{1,\ldots,n_y\}$ denote a specific output dimension with a kernel function denoted with $k^m(\cdot,\cdot)$. 
	Then the GPR posterior predictions for a new feature $z$ are given by
%	\vspace{-0.1cm}
	\begin{align*}
	%\label{eq:gpr}
	\small
	\mu^m(z|D) &= K_{zZ}^m[K_{ZZ}^m+\sigma_{m}^2 \mathcal{I}]^{-1}[\Delta]_{\cdot m} \\
	\mathit{cov}^m(z|D) &= K_{zz}^m-K_{zZ}^m[K_{ZZ}^m+\sigma_{m}^2 \mathcal{I}]^{-1}K_{Zz}^m,
	\end{align*}
	where $K_{Zz}$ denotes a matrix of covariances 
	with $[K_{ZZ}]_{ij} = k(z^i, z^{j})$, $[\Delta]_{\cdot m}$ denotes the column of $\Delta$ for the output channel $m$,  $\sigma_{m}^2$ denotes the noise hyper-parameter of the GP, and we assume zero prior mean without loss of generality.
	The remaining matrices $K_{zZ}$, $K_{Zz}$, and $K_{zz}$ are evaluated similarly (e.g., see~\cite{hewing2019cautious}).
	Here, we utilize the squared-exponential kernel for all output channels,
	\begin{align}
	\label{eq:se}
	k^m(z^i,z^j) = \sigma_{s,m}^2 \exp\{-\textstyle{\frac{1}{2}}||z^i-z^j||^2_{\Lambda^{-1}_m} \},
	\end{align}
	where $\sigma_{s,m}$ and $\Lambda_m$ are the variance and length-scale hyper-parameters of the kernel, respectively. 
	\revn{Note that other kernels may be utilized instead.}
	The hyper-parameters $\theta^m =\{\sigma_{m}, \sigma_{s,m}, \Lambda_m \}$ for each channel are learned from training data by maximizing the log-likelihood of the observations, in an off-line preparation step.
	Note that learning $\sigma_{m}$ effectively corresponds to learning the noise variance of (\ref{eq:dynSys}).
	The learned hyper-parameters are identically fixed throughout each iteration and are unchanged between the iterations.

	Combining the GPs of each output channel, the resulting GP \revn{posterior distribution} for estimating the unknown dynamics $g$ at the test point $z$, given the data $D$ is  
	\begin{align}
	\label{eq:gp_pred}
	p(z|D) \sim \mathcal{N}(\mu(z|D),\sigma(z|D)),
	\end{align}
	where $\mu(z|D) = [\mu^1(z|D),\ldots,\mu^{n_y}(z|D)]^T$, and $\sigma(z|D) = \mathop{diag}\left(\mathit{cov}^1(z^*|D),\ldots,\mathit{cov}^{n_y}(z|D)\right)$, where $\mathop{diag}$ denotes a block diagonal matrix.
	The resulting output dynamics prediction including the GP approximation is given as
	\begin{align}
	\label{eq:ogpr}
	y_k^p(t+1) = h(x_k(t), u_k(t)) + p(x_k(t), u_k(t)|D),
	\end{align} 
	where \revn{$y_k^p$ is the predicted output trajectory, and (\ref{eq:ogpr}) is named as the output GPR (OGPR) method.}
	Similarly we denote $y^{\mu}_k$ as the mean prediction of (\ref{eq:ogpr}).
	OGPR simplifies our model complexity and enables us to predict in the lifted domain (\ref{eq:io_map}).
	Additionally, the future output predictions do not depend on the current output prediction, eliminating the need for uncertainty propagation.

	The dataset $D$ has a great effect on the quality (e.g., accuracy and uncertainty) of the GPR predictions. 
	In general, 
	whenever $z$ is close to a feature in $Z$, i.e. $\min_{\tilde{z}\in Z}||z - \tilde{z}||$ small, we have higher confidence in the predictive covariance.
	Our goal is to update the dataset $D$ after each run to improve the quality of the GPR predictions by utilizing the most relevant datapoints.
	We provide further remarks on the dataset update in later sections.

	\section{Run-to-Run Learning-Based Controller}
	\label{sec:controller}
	After each iteration, we evaluate the control action for the next iteration by utilizing the input and the measurements \revn{from} the current iteration. 
	The iterative process is given by 
	\begin{align*}
	\uvec_{k+1} = \Pi(\uvec_k, \yvec_k | D_k),
	\end{align*}
	where we utilize the dataset $\textstyle{D_k\!\in\!\mathcal{D}}$
	at iteration $k$ 
	to evaluate a new control action via the run-to-run control map $\textstyle{\Pi\!: \reals^{n_un_i}\!\times\!\reals^{n_yn_i}\!\times\!\mathcal{D}\to\!\reals^{n_un_i}}$. 
	We start by providing the update equations of the proposed controller, proceed with a preliminary analysis, and provide implementation details. 
	
	\subsection{Update evaluation}
	
	The proposed controller updates for the map $\Pi$ for each iteration are given by the following two steps
	\begin{subequations}
		\label{eq:r2r_updates}
		\begin{align}
		\vvec_{k+1} &= \uvec_k + W^{-1}(H^TQ(\evec_k) - S\uvec_k)	\label{eq:grad_step},\\
		\uvec_{k+1} &= \argmin_{\uvec\in\mathcal{U}}\{J(\yvec^p(\uvec)|D_{k})+\textstyle{\frac{1}{2\lambda}}||\uvec-\vvec_{k+1}||^2\} \label{eq:prox_step},
		\end{align}
	\end{subequations}
	where we have $\textstyle{\evec_{k} = \rvec-\yvec_k}$, $\textstyle{W = H^TQH+R+S}$ with $Q,R,S$ symmetric positive definite weight matrices of appropriate dimensions, $\textstyle{\lambda>0}$ is the regularization weight, $\mathcal{U}\subseteq\reals^{n_un_i}$ is a convex constraint set for the input, and finally $J(\yvec^p(\uvec)|D_{k})$ is a cost function for the contour error of the $\yvec^p(\uvec)$ predicted by the GP for input $\uvec$; details are provided below.
	The weight $S$ penalizes the control effort, while $R$ penalizes the variation of the input signal between iterations.
	We utilize the nominal lifted model to evaluate an approximate input signal in the evaluation step (\ref{eq:grad_step}) and utilize the output mismatch model to improve the input in (\ref{eq:prox_step}).
	The step (\ref{eq:grad_step}) is a norm-optimal ILC update (e.g.,~\cite{barton2010norm}), which is interconnected with (\ref{eq:prox_step}). 
	
	\subsubsection{Cost function}
	We denote $\bm{\epsilon}(\bm{y})$, a general purpose error penalty, which may be designed based on the application of interest (e.g. the individual axis or contour error discussed above). 
	We then define
	\begin{align}
	\label{eq:lifted_cost}
	J(\yvec^p(\uvec)|D) = ||\bm{\epsilon}(\yvec^{\mu}(\uvec))||^2_{Q_{\epsilon}}+ \Sigma(\uvec|D)
	\end{align}
	where $\bm{\epsilon}(\yvec^{\mu}(\uvec))$ is an error metric for the mean prediction 
	and $\textstyle{\Sigma(\uvec|D) = \sum_t tr(\tilde{Q}_{\epsilon}y^{\Sigma}(t))}$ denotes the sum of the traces of covariances for the predicted trajectory with $\tilde{Q}_{\epsilon}$ denoting the corresponding weights from $Q_{\epsilon}$.
	The cost function penalizes the tracking errors and the uncertainty over the output mismatch predictions. 

	\subsubsection{Dataset update}
	The dataset $D_k$, \revn{used in (\ref{eq:lifted_cost}),} is re-evaluated at each iteration to improve the predictive performance of the GPR.
	We use the inputs and measurements from the current iteration to evaluate the features $z_k(t) =[x_k(t)^T,u_k(t)^T]^T$ and the mismatch observations $\delta_k(t) = y_k(t+1) - h(z_k(t))$ as in (\ref{eq:delta}).
	
	The choice of dataset update rule depends on the expected change in process mismatch between runs and computational resources available.
	For rapidly changing processes, old data may deteriorate prediction accuracy; whereas for a slowly changing (or stationary) process in the R2R domain, keeping the data from a few or all past runs may be desirable.
	However as the dataset gets larger, the computational cost of predictions increases as well, thus a balanced design must be evaluated for the specific application at hand.
	Here, we adopt a dataset update rule that uses \textit{only} the data from the most recent run for GPR predictions.
	Loosely speaking, due to the repetitive nature of the problem, the data sampled in the current iteration is expected to provide a high-performing dataset for predicting the mismatch in the next iteration.
	The dataset is then given by 
	\begin{align}
	\label{eq:dataset_update}
	D_k = (Z_k, \Delta_k),
	\end{align}
	where we have the features $\textstyle{Z_k\!=\![z_k(0),\ldots,z_k(n_i\!-\!1)]^T}$, and the observations  $\textstyle{\Delta_k=[\delta_k(0),\ldots,\delta_k(n_i\!-\!1)]^T}$. 
	Convergence properties of (\ref{eq:grad_step})-(\ref{eq:prox_step}) with the update rule (\ref{eq:dataset_update}) are studied in the next section.

	\subsection{Analysis}
	\label{sec:analysis}
	We state the convergence properties of (\ref{eq:r2r_updates}) by analyzing the two steps (\ref{eq:grad_step}) and (\ref{eq:prox_step}) individually.
	\subsubsection{Update (\ref{eq:grad_step})}
	\sout{First,} We assume that the true operator of (\ref{eq:io_map}), denoted by $H_t$ lies in a \revn{unstructured norm-bounded} uncertainty set around the nominal dynamics model $H$.
	Following~\cite{barton2010norm}, we consider the optimization problem
	\vspace{-0.05cm}
	\begin{align}
	\label{eq:ia_error}
	\min\limits_{\bm{\phi}}\{ ||\bm{e}_{k+1}(\bm{\phi})||^2_Q + ||\bm{\phi}||^2_S +  ||\bm{\phi}-\bm{u}_{k}||^2_R\},
	\end{align}
	where we have the error model $\textstyle{\bm{e}_{k+1}(\bm{\phi}) = \bm{r}-(H\bm{\phi}+\bm{w})}$ for an input $\bm{\phi}$, and 
	the previous control input $\bm{u}_k$.
	If we \revn{adopt a common assumption} that the disturbance $\bm{w}$ is approximately repeatable in the lifted R2R domain during the experiments, then by utilizing the tracking error from measurements, step (\ref{eq:grad_step}) in isolation (i.e., $\uvec_{k+1}$ set as $\bm{v}_{k+1}$) converges \revn{monotonically} to a fixed point. 
	\revn{Thus, (\ref{eq:grad_step}) in isolation utilizes measurements to minimize (\ref{eq:ia_error}) by utilizing the assumption on the relationship between the model $H$ and $H_t$.}
	Refer to~\cite{barton2010norm} for further details \revn{and derivations}.
	Therefore, we have the step (\ref{eq:grad_step}) as a contraction.
	The fixed point of the robust ILC iteration (\ref{eq:grad_step}) is an approximation of the true fixed point due to model mismatch.
	Thus, we utilize the mismatch compensation through the GPR predictions to improve the steady state error 
	in (\ref{eq:prox_step}).
	
	\subsubsection{Update (\ref{eq:prox_step})}

	Under the assumption that the GPR hyper-parameters are \revn{accurate approximations of} the true hyper-parameters \revn{(so that the GPR predictions are accurate),} we can predict the dataset for the next iteration (i.e., predicting for $\textstyle{k+1}$ at run $k$) by utilizing the mean output predictions of (\ref{eq:ogpr}), given low observation noise.
	The application domain of precision motion control often employs high-performance sensing, which justifies the low noise assumption.
	Consider the update (\ref{eq:prox_step}) in isolation (i.e., without the step (\ref{eq:grad_step})), rewritten in the variable $\xi$ as
	\vspace{-0.1cm}
	\begin{align}
	\label{eq:prox_only}
	\hat{\xi}_{k+1}= \argmin\limits_{\xi\in\mathcal{U}}\{J(\yvec^p(\xi))|D) + \textstyle{\frac{1}{2\lambda}}||\xi\!-\! \hat{\xi}_k||^2 \}.
	\end{align}
	We state the following properties for (\ref{eq:prox_only}) and then sketch the proof of convergence for the combined (\ref{eq:grad_step})-(\ref{eq:prox_step}).
	\begin{proposition}
		\label{prop:bounded_cost}
%		We have
		$\textstyle{J(\yvec^p(\hat{\xi}_{k+1})|D_{k+1}) \leq J(\yvec^p(\hat{\xi}_{k+1})|D_{k})}$ holds for the dataset update (\ref{eq:dataset_update}) in expectation over the predicted observations given  the true hyper-parameters ($\bm{\theta}$) and low observation noise.
	\end{proposition}
	\begin{proof}
		Using the output prediction in (\ref{eq:ogpr}) with the true $\bm{\theta}$ we predict the observed data in the next iteration ($k+1$) to
		be $\yvec^{\mu}(\hat{\xi}_{k+1})$ in expectation, from which the observations $\Delta_{k+1}$ for the dataset $D_{k+1}$ follows.
		Since both costs are evaluated with the same input, $\hat{\xi}_{k+1}$, the cost $J(\yvec^p(\hat{\xi}_{k+1})|D_{k+1})$ has predictions at the locations of the features in $D_{k+1}$.
		Thus, we approximate the predicted expectation to be the observations with low noise and get the error of the expected trajectory $\textstyle{||\bm{\epsilon}(\yvec^{\mu}(\hat{\xi}_{k+1}))||^2_{Q_{\epsilon}}}$ equal on both costs. 
		Following the same reasoning and noting that the prediction uncertainty (i.e., the covariance matrix) is a function of the locations of the features 
		we have $\Sigma(\hat{\xi}_{k+1} | D_{k+1}) \leq \Sigma(\hat{\xi}_{k+1} | D_{k})$.
		Combining the stated arguments we arrive at the desired result.
	\end{proof}
	\begin{proposition}
		\label{prop:prox}
		The optimal value of the cost function 
		in (\ref{eq:prox_only}) is nonincreasing in the iteration domain in expectation over the predicted observations such that
		\begin{align*}
		J(\yvec^p(\hat{\xi}_{k+1})|D_{k}) - J(\yvec^p(\hat{\xi}_{k})&|D_{k-1}) \leq -\textstyle{\frac{1}{2\lambda}}||\hat{\xi}_{k+1}-\hat{\xi}_{k}||^2.
		\end{align*}
	\end{proposition}
	\begin{proof}
%		We first note that 
		Due to the regularization term $||\xi\!-\! \hat{\xi}_k||^2$ in (\ref{eq:prox_only}), for the minimizer $\hat{\xi}_{k+1}$ we have that 
		\begin{align*}
		J(\yvec^p(\hat{\xi}_{k+1})|D_{k}) +\textstyle{\frac{1}{2\lambda}}||\hat{\xi}_{k+1}-\hat{\xi}_k||^2 &\leq J(\yvec^p(\hat{\xi}_{k})|D_{k}), \\		
		&\leq J(\yvec^p(\hat{\xi}_{k})|D_{k-1}),
		\end{align*}
		where we used the optimality condition for the optimization in (\ref{eq:prox_only}), and Proposition~\ref{prop:bounded_cost}.
	\end{proof}

	We only sketch the convergence proof of (\ref{eq:grad_step})-(\ref{eq:prox_step}) due to space limitations. 
	The results outlined in the previous section imply that the updates in (\ref{eq:r2r_updates}) are the interconnection of (\ref{eq:grad_step}), a contraction, and (\ref{eq:prox_step}), an optimization step with nonincreasing cost in the iteration domain. 
	Observe that whenever both (\ref{eq:grad_step}) and (\ref{eq:prox_step}) share the same term in their cost function, i.e., $||\bm{\epsilon}(\yvec^{\mu}(\uvec))||^2_{Q_{\epsilon}}$, Proposition~\ref{prop:prox} holds for the interconnection whenever $||\vvec_{k+1}-\uvec_k||\leq \beta$, for some $\beta>0$.
	Then, by utilizing an invariance argument (e.g., LaSalle's invariance principle) 
	for the interconnection, it can be shown that the system response under the updates in (\ref{eq:grad_step})-(\ref{eq:prox_step}) is convergent.
	The fixed point of 
	the interconnection 
	improves the steady state error of the step (\ref{eq:grad_step}) alone, since (\ref{eq:prox_step}) optimizes for the true model of the plant, given the dataset $D_k$.
	We illustrate this improvement 
% 	on the steady state error 
	in the case study.

	\subsection{Implementation details for 2D contouring control}
	
	In this section, we focus on the case when $\bm{\epsilon}$ denotes the contouring error. 
	Following \cite{liniger2019real}, we utilize a path parameter $\textstyle{s(t+1) = s(t) + T\varphi(t)}$ to approximate $\hat{s}$, where $T$ is the sampling time and $\varphi(t)$ is the velocity along the path.
	Introduction of $s(t)$ inherently creates a lag-error $e^{\ell}(t)$ resulting from the mismatch between the true minimizer arc parameter $\hat{s}(t)$ and $s(t)$ illustrated in Fig.~\ref{fig:errors}.
	A linear approximation of the true contour error $\hat{e}^c(t)$ denoted as $e^c$ and lag error $e^{\ell}$ along the path with respect to $s(t)$ is 
	\begin{subequations}
		\label{eq:errors}
% 		\small
		\begin{align}
% 		\small
			e^{\ell}&(t+1) = r'_1(s(t))(r_1(s(t)) -  y_1(t+1)) \notag \\
			&+r'_2(s(t))(r_2(s(t)) -  y_2(t+1)) + T\varphi(t),\label{eq:lagerror} \\
			e^c&(t+1) = -r'_2(s(t))(r_1(s(t)) -  y_1(t+1)) \notag \\
			&+r'_1(s(t))(r_2(s(t)) -  y_2(t+1)), \label{eq:conerror}
		\end{align}
	\end{subequations}
	where the path variable is $\bm{r}(s) = [r_1(s), r_2(s)]^T$, $\bm{r}'(s) = [r'_1(s), r'_2(s)]^T$ is the normalized path derivative,
	and the subscripts $1$ and $2$ denote the axes of the 2D 
	system.
	\begin{figure}[h]
		\centering
		\includegraphics[width=0.4\columnwidth]{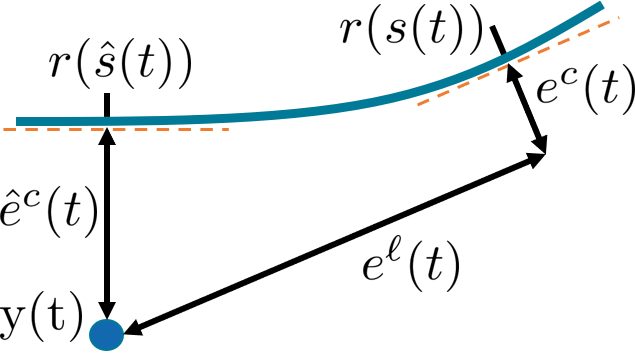}
		\vspace{-0.2cm}
		\caption{Illustration of the contour and lag error approximations.}
		\label{fig:errors}
				\vspace{-0.35cm}
	\end{figure}
	
	To improve the accuracy of the contour error approximation, we require $e^{\ell}$ to be very small, thus we apply a higher penalty than $e^c$.
	The linearized error metric is \revn{then} given as
	\begin{align}
		\label{eq:epsilons}
		&\tilde{\epsilon}(t+1) =\! \Phi(s(t))[y^{\mu}(t+1)^T,\varphi(t)]^T\! +\! c(s(t)),
	\end{align}
	where $\tilde{\epsilon}(t) = [e^{\ell}(t), e^c(t)]^T$, with $\Phi(s(t))\in\reals^{2\times 3}$ and $c(s(t))\in\reals^{2}$ are the corresponding linear and affine terms from (\ref{eq:lagerror})-(\ref{eq:conerror}), and $y^{\mu}(t)$ is the mean prediction of system output (\ref{eq:ogpr}), given the nominal state and input.
	
	The optimization step in (\ref{eq:prox_step}) is a nonlinear and nonconvex output control problem (OCP). 
	Utilizing the GP approximation of the output map, we can write the OCP as
	\vspace{-0.1cm}
	\begin{align*}
	\min\limits_{\uvec\in\mathcal{U},\bm{\varphi}\in\mathcal{P}}~&  \textstyle{\frac{1}{2\lambda}}||\uvec-\vvec_{k+1}||^2\! +\!\textstyle{\sum_{t=1}^{n_i}} ||\tilde{\epsilon}(t)||^2_{Q_{e}}\!+\!tr(Q_{e}y^{\Sigma}(t))\\
	\text{s.t.:}~& 
	y^p(t+1) = h(x(t), u(t)) + p(x(t), u(t)|D_k) \\
	&x(0) = x_*(0),\quad t = 0,\ldots,\textstyle{(n_i-1)}, \\
	& s(t+1) = s(t) + T\varphi(t), \quad s(t) \in [0,1]
	\end{align*}
	where $\mathcal{P}$ is the path velocity constraint set \revn{and the rest of the terms are introduced earlier}. 
	To improve the computational performance, following~\cite{liniger2019real,rupenyan2021performance} we resort to a receding horizon optimization procedure to solve the OCP.
	The resulting receding horizon problem is given by
		\vspace{-0.1cm}
		\begin{align}
		\min\limits_{\tilde{\uvec}\in\tilde{\mathcal{U}},\tilde{\bm{\varphi}}\in\tilde{\mathcal{P}}}& \tilde{J}_f(\tilde{\epsilon}(\bar{N}) | D )\!+\!\textstyle \textstyle{\frac{1}{2\lambda}}||\tilde{u}(\bar{N})-v_{k+1}(\bar{N})||^2 \notag \\
		&+\textstyle\sum_{t=t'}^{\bar{N}-1}\! \tilde{J}_t(\tilde{\epsilon}(t) | D )\!+\!\textstyle{\frac{1}{2\lambda}}||\tilde{u}(t)\!-\!v_{k+1}(t)||^2\notag \\
		\text{s.t.: } & y^p(t+1) = h(x(t), \tilde{u}(t)) + p(x(t), \tilde{u}(t)|D) \notag \\
		&x(t') = \bar{x}(t'),s(t') = \bar{s}(t')\quad D = D_k(a,b,\alpha),\notag \\
		& s(t+1) = s(t) + T\tilde{\varphi}(t), \quad s(t) \in [0,1] \label{eq:mpcc}
		\end{align}
	where the constraint sets $\tilde{\mathcal{U}}$ and $\tilde{\mathcal{P}}$ are adjusted for the horizon length, and $\bar{x}(t'), \bar{s}(t')$ are the initial conditions at each $t'$.
	We use the shorthand $\bar{N}\! =\! t'\!+\!N$ with the current time $t'$ and the horizon length $N$, and define the cost function
	\begin{align*}
	\tilde{J}_t(\tilde{\epsilon}(t) | D) &= ||\tilde{\epsilon}(t)||^2_{Q_e} \!+\!tr(Q_e y^{\Sigma}(t))- \gamma\varphi(t) + || u(t)||^2_{Q_u},
	\end{align*}
	\revn{with $\textstyle{\gamma \in\reals_+}$, and positive semidefinite weights $\textstyle{Q_u,Q_e}$,} 
	where the dependency on $D$ has been omitted to simplify the notation.
	We define the dataset 
	\begin{align}
	D_k(a,b,\alpha) = \{(z_k(\ell),\delta_k(\ell)~|~ \ell = [\alpha+a, \alpha+b] \}.
	\end{align} 
	While utilizing the receding horizon computation for the OCP, we also utilize a fixed window size of data points from the dataset $D_k$ to improve the computational performance of the GPR within the optimization problem (\ref{eq:mpcc}).
	A common choice for updating the dataset $D$ when there is no previous iteration data is to use $\textstyle{\alpha = t'}$, $\textstyle{a = -n_w-1}$, and $b = -1$, with $n_w$ as the window size.
	However, since we have access to full data from the previous iteration (see (\ref{eq:dataset_update})), we instead utilize $\alpha = \arg\min_j \{|s(t') - s_{k}(j)|\}$ 
	with $\bm{s}_{k}=[s_k(1),\ldots,s_k(n_i)]$ as the vector of path parameter values from the previous iteration, and we take $a = -n_w/2$, $b = n_w/2$ for improved prediction performance, while keeping the learned hyper-parameters constant for the whole geometry on all iterations as explained in Section~\ref{sec:gpr}.

	By solving (\ref{eq:mpcc}) in a receding horizon fashion and denoting the first optimal input of each step by, $u^*(t')$, we evaluate the optimal control input $\textstyle{\uvec_{k+1} = [u^*(0),\ldots,u^*(n_i\!-\!1)]^T}$ for the next iteration in (\ref{eq:r2r_updates}).
	The implementation of the proposed R2R controller is given in Algorithm~\ref{alg:r2r_algo}.
	\vspace{-0.2cm}
		\begin{algorithm}[h]
			
		\caption{R2R control with mismatch compensation}
		\label{alg:r2r_algo}
		\begin{algorithmic}[1]
			\State \textbf{Input:} $\uvec_0$, $\theta$, $Q,R,S,T,W,r,Q_{\epsilon},\lambda, \rvec$, 
			\State $k\gets 0$, $\uvec_{k}\gets \uvec_0$
			\While {$||\uvec_{k+1} - \uvec_k|| \geq \epsilon$ or $k \leq k_{max}$}
			\State Apply $\uvec_k$ to the real system (\ref{eq:io_map}) and measure $\yvec_k$.
			\State Evaluate $\vvec_{k+1}$ via (\ref{eq:grad_step}).
			\State Update the dataset $D_k$ as given in (\ref{eq:dataset_update}).
			\State Solve (\ref{eq:mpcc}) in a receding horizon and evaluate $\uvec_{k+1}$.
			\State $k\gets k+1$
			\EndWhile
		\end{algorithmic}	
	\end{algorithm}
	\vspace{-0.5cm}
	\section{Performance Results}
	\label{sec:caseStudy}

	Here we present a simulation study and preliminary experimental results to illustrate the \revn{controller} performance. 
	
	\subsection{Setup}
	For the simulation study, we identify two linear models from experimental data; a high-fidelity true system model to represent the true input-output map of the system in (\ref{eq:dynSys}), and a low-fidelity nominal controller model for (\ref{eq:linSys}).
	The nominal model is a single input single output (SISO) model for each axis (position input, position output) and identical for both axes for simulation purposes, resulting in a two-input two-output system.
	\revn{Noise with zero mean and $1$ micron standard deviation is added to the output measurements in simulation.}

	The experimental setup in Fig.~\ref{fig:testbed} was used for preliminary experimental results. 
	\revn{Input} trajectories \revn{defining position references for the machine} at each iteration are computed by the controller and pushed to the machine. 
	The position measurements \revn{of both axes are reported at the end of each run.}
	Further details about the low-level controllers and the setup of the testbed are presented in~\cite{haas2019mpcc}.
	The system has two axes with inconsistent dynamics resulting in high model mismatch, especially at corner geometries such as the ones in the octagonal reference.
	\revn{2-state} SISO nominal models per axis are obtained from experimental data.

	In the case study, we  \revn{also} utilize an approximation of the OGPR term in (\ref{eq:mpcc}), which we denote with \textit{A-GPR}.
	In the approximation, we utilize \revn{a} previous solution trajectory of (\ref{eq:mpcc}) to predict the output mismatch term in (\ref{eq:epsilons}) \revn{as a} constant during the optimization, resulting in a quadratic program. 
	This approximation is reasonable whenever solution trajectories do not vary too much between consecutive \revn{solutions} of (\ref{eq:mpcc}).
	We incorporate additional state constraints in (\ref{eq:mpcc})
	to force the solution trajectories to be close to their approximation points. 
	
	\subsection{Results}

	\subsubsection{Simulation} We compare the convergence performance of several controllers. 
	The \textit{feed-forward controller} (FF) uses only the nominal linear model of the system. 
	The ILC uses only the update (\ref{eq:grad_step}) without mismatch compensation, penalizing the individual axis error.
	The \textit{A-GPR} controller uses \revn{the approximation described in the previous section.} 
	Finally, the controller denoted by \textit{F-GPR} solves \revn{(\ref{eq:mpcc}) with the full GPR model.} 
	We initialize the controllers with the $\uvec_0$ corresponding to the reference geometry itself shown in Fig.~\ref{fig:testbed}.
	\begin{figure}[t]
		\centering
		\includegraphics[width=0.75\columnwidth]{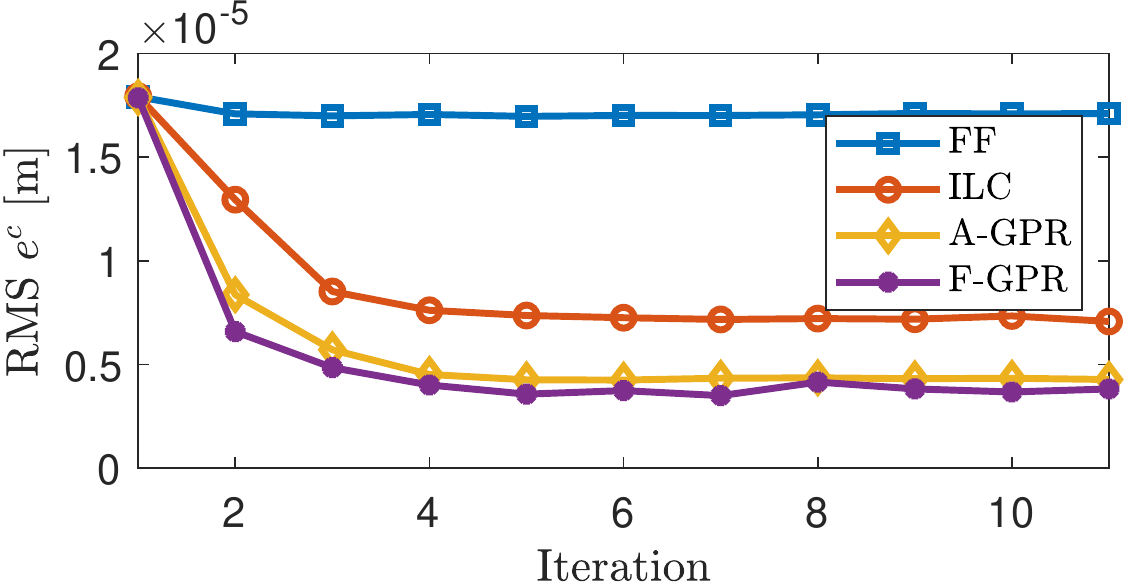}
		\vspace{-0.2cm}
		\caption{Comparison of various controllers for the simulation results.}
		\label{fig:sim}
				\vspace{-0.65cm}
	\end{figure}
	The simulation results are shown in Fig.~\ref{fig:sim} and the approximate percent improvements of the converged controllers are in Table~\ref{tab:simresults}. 
	We see that the FF controller leads to minimal improvement due to the model mismatch. 
	The Full-GPR solutions show the best result by improving the RMS error by $79\%$ overall with the A-GPR controller performing comparably. 
	With our current implementation on MATLAB using \texttt{quadprog} and \texttt{fmincon} as solvers requires $30.6\pm 1.4$ seconds for A-GPR with $N=50$ and $876\pm 67$ seconds for the F-GPR with $N=10$ 
	for a single iteration on a PC with Intel i7-7700HQ CPU.
	We note that higher horizon length for F-GPR may further improve the controller performance \revn{at a higher computational cost}.
	
	\renewcommand{\arraystretch}{0.05}
	\begin{table}[h]
		\label{tab:simresults}
		\vspace{-0.3cm}
		\caption{Percent improvement results for Fig.~\ref{fig:sim}.}
		\vspace{-0.5cm}
		\begin{center}
			\begin{tabular}{@{} c c c c@{}}
				\toprule
				\makebox[1cm]{FF} & \makebox[1cm]{ILC} & \makebox[1cm]{A-GPR}&\makebox[1cm]{F-GPR}\\
				\cmidrule(lr){1-4}
				$4.8\%$ &$59.6\%$ & $75.9\%$ & $79.0\%$\\
				\bottomrule
			\end{tabular}
		\end{center}
		\vspace{-0.5cm}
	\end{table}
	\renewcommand{\arraystretch}{1}

	\subsubsection{Experimental}
	For the experimental results, we utilize the A-GPR controller with $N=50$.

	\begin{figure}[h]
		\vspace{-0.28cm}
		\centering
		\includegraphics[width=0.7\columnwidth]{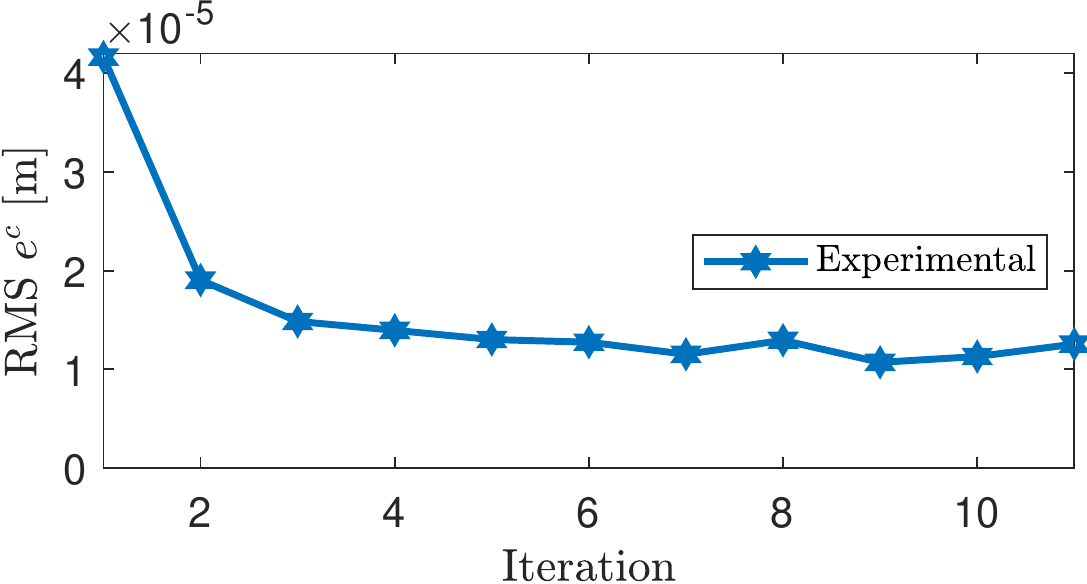}
		\vspace{-0.2cm}
		\caption{Experimental results of the A-GPR R2R controller. }
		\label{fig:exp}
		\vspace{-0.45cm}
	\end{figure}
	Figure~\ref{fig:exp} illustrates the results of A-GPR on the experimental setup.
	We see an approximately $70.9\%$ improvement of the RMS when compared to the initial RMS level in the experimental setup. 
	The experimental results illustrate that the proposed approach is suitable for industrial applications and provides promising results for future research.
	
	\section{Conclusion and Future Work}
	\label{sec:conc}
	In this work, we present a novel run-to-run controller that utilizes machine learning methods to model the output mismatch of the nominal system model to improve tracking performance.
	The results demonstrate that the proposed controller works in practice and warrants future studies with optimized implementations for industrial settings.
	The proposed method may be applied to other domains such as robotic motion control and additive manufacturing applications with minimal modifications.
	Future work will focus on presenting the formal results on performance and a demonstration of the process on similar experimental setups by incorporating varying levels of model mismatch, imperfect learning, and output constraints to meet manufacturing requirements.

	\bibliographystyle{IEEEtran}
	\bibliography{r2r_bib2}

% Generated by IEEEtran.bst, version: 1.14 (2015/08/26)
\begin{thebibliography}{10}
\providecommand{\url}[1]{#1}
\csname url@samestyle\endcsname
\providecommand{\newblock}{\relax}
\providecommand{\bibinfo}[2]{#2}
\providecommand{\BIBentrySTDinterwordspacing}{\spaceskip=0pt\relax}
\providecommand{\BIBentryALTinterwordstretchfactor}{4}
\providecommand{\BIBentryALTinterwordspacing}{\spaceskip=\fontdimen2\font plus
\BIBentryALTinterwordstretchfactor\fontdimen3\font minus
  \fontdimen4\font\relax}
\providecommand{\BIBforeignlanguage}[2]{{%
\expandafter\ifx\csname l@#1\endcsname\relax
\typeout{** WARNING: IEEEtran.bst: No hyphenation pattern has been}%
\typeout{** loaded for the language `#1'. Using the pattern for}%
\typeout{** the default language instead.}%
\else
\language=\csname l@#1\endcsname
\fi
#2}}
\providecommand{\BIBdecl}{\relax}
\BIBdecl

\bibitem{umlauft2018uncertainty}
J.~Umlauft, L.~P{\"o}hler, and S.~Hirche, ``An uncertainty-based control
  {Lyapunov} approach for control-affine systems modeled by {Gaussian}
  process,'' \emph{IEEE Control Systems Letters}, vol.~2, no.~3, pp. 483--488,
  2018.

\bibitem{oomen2018advanced}
T.~Oomen, ``Advanced motion control for precision mechatronics: Control,
  identification, and learning of complex systems,'' \emph{IEEJ Journal of
  Industry Applications}, vol.~7, no.~2, pp. 127--140, 2018.

\bibitem{wang2017newton}
Z.~Wang, C.~Hu, Y.~Zhu, S.~He, M.~Zhang, and H.~Mu, ``{Newton-ILC} contouring
  error estimation and coordinated motion control for precision multiaxis
  systems with comparative experiments,'' \emph{IEEE Transactions on Industrial
  Electronics}, vol.~65, no.~2, pp. 1470--1480, 2017.

\bibitem{altin2014robust}
B.~Alt{\i}n and K.~Barton, ``Robust iterative learning for high precision
  motion control through {L1} adaptive feedback,'' \emph{Mechatronics},
  vol.~24, no.~6, pp. 549--561, 2014.

\bibitem{amann1998predictive}
N.~Amann, D.~H. Owens, and E.~Rogers, ``Predictive optimal iterative learning
  control,'' \emph{International Journal of Control}, vol.~69, no.~2, pp.
  203--226, 1998.

\bibitem{colombino2019towards}
M.~Colombino, J.~W. Simpson-Porco, and A.~Bernstein, ``Towards robustness
  guarantees for feedback-based optimization,'' in \emph{2019 IEEE 58th
  Conference on Decision and Control (CDC)}.\hskip 1em plus 0.5em minus
  0.4em\relax IEEE, 2019, pp. 6207--6214.

\bibitem{baumgartner2020zero}
K.~Baumg{\"a}rtner and M.~Diehl, ``Zero-order optimization-based iterative
  learning control,'' in \emph{2020 59th IEEE Conference on Decision and
  Control (CDC)}.\hskip 1em plus 0.5em minus 0.4em\relax IEEE, 2020, pp.
  3751--3757.

\bibitem{cheah2006adaptive}
C.-C. Cheah, C.~Liu, and J.-J.~E. Slotine, ``Adaptive tracking control for
  robots with unknown kinematic and dynamic properties,'' \emph{The
  International Journal of Robotics Research}, vol.~25, no.~3, pp. 283--296,
  2006.

\bibitem{modares2014optimal}
H.~Modares and F.~L. Lewis, ``Optimal tracking control of nonlinear
  partially-unknown constrained-input systems using integral reinforcement
  learning,'' \emph{Automatica}, vol.~50, no.~7, pp. 1780--1792, 2014.

\bibitem{beckers2019stable}
T.~Beckers, D.~Kuli{\'c}, and S.~Hirche, ``Stable gaussian process based
  tracking control of euler--lagrange systems,'' \emph{Automatica}, vol. 103,
  pp. 390--397, 2019.

\bibitem{hewing2019cautious}
L.~Hewing, J.~Kabzan, and M.~N. Zeilinger, ``Cautious model predictive control
  using {Gaussian} process regression,'' \emph{IEEE Transactions on Control
  Systems Technology}, vol.~28, no.~6, pp. 2736--2743, 2019.

\bibitem{berkenkamp2015safe}
F.~Berkenkamp and A.~P. Schoellig, ``Safe and robust learning control with
  {Gaussian} processes,'' in \emph{2015 European Control Conference
  (ECC)}.\hskip 1em plus 0.5em minus 0.4em\relax IEEE, 2015, pp. 2496--2501.

\bibitem{li2020data}
X.~Li, H.~Zhu, J.~Ma, T.~J. Teo, C.~S. Teo, M.~Tomizuka, and T.~H. Lee,
  ``Data-driven multiobjective controller optimization for a magnetically
  levitated nanopositioning system,'' \emph{IEEE/ASME Transactions on
  Mechatronics}, vol.~25, no.~4, pp. 1961--1970, 2020.

\bibitem{bai2008towards}
E.-W. Bai and J.~Reyland~Jr, ``Towards identification of wiener systems with
  the least amount of a priori information on the nonlinearity,''
  \emph{Automatica}, vol.~44, no.~4, pp. 910--919, 2008.

\bibitem{koren1991variable}
Y.~Koren and C.-C. Lo, ``Variable-gain cross-coupling controller for
  contouring,'' \emph{CIRP annals}, vol.~40, no.~1, pp. 371--374, 1991.

\bibitem{barton2010norm}
K.~L. Barton and A.~G. Alleyne, ``A norm optimal approach to time-varying {ILC}
  with application to a multi-axis robotic testbed,'' \emph{IEEE Transactions
  on Control Systems Technology}, vol.~19, no.~1, pp. 166--180, 2010.

\bibitem{liniger2019real}
A.~Liniger, L.~Varano, A.~Rupenyan, and J.~Lygeros, ``Real-time predictive
  control for precision machining,'' in \emph{2019 IEEE 58th Conference on
  Decision and Control (CDC)}.\hskip 1em plus 0.5em minus 0.4em\relax IEEE,
  2019, pp. 7746--7751.

\bibitem{rupenyan2021performance}
A.~Rupenyan, M.~Khosravi, and J.~Lygeros, ``Performance-based trajectory
  optimization for path following control using bayesian optimization,'' in
  \emph{2021 IEEE 60th Conference on Decision and Control (CDC)}.\hskip 1em
  plus 0.5em minus 0.4em\relax IEEE, 2021.

\bibitem{haas2019mpcc}
T.~Haas, S.~Weikert, and K.~Wegener, ``{MPCC-Based} set point optimisation for
  machine tools,'' \emph{International Journal of Automation Technology},
  vol.~13, no.~3, pp. 407--418, 2019.

\end{thebibliography}

\end{document}